\documentclass[12pt]{article}
\usepackage{graphicx}
\usepackage{amssymb}
\usepackage{amsfonts}
\usepackage{amsthm}
\usepackage{amsmath}
\usepackage{a4wide}

\newtheorem{proposition}{Proposition}
\newtheorem{theorem}{Theorem}
\newtheorem{algorithm}{Algorithm}
\newtheorem{definition}{Definition}

\newtheorem{example}{Example}

\newcommand{\bx}{\mathbf{x}}
\newcommand{\dfdx}{\frac{\partial f(\mathbf{x})}{\partial x_j}}
\newcommand{\flag}{\textrm{flag}}

\renewcommand\footnotemark{}

\begin{document}

\title{Nested Canalyzing Depth and Network Stability}

\author{Lori Layne \thanks{\hspace{-0.2in}$^*$ Department of
    Mathematical Sciences, Clemson University, Clemson, SC 29634} \and
  Elena Dimitrova \and Matthew Macauley}

\maketitle

\begin{abstract}
  We introduce the \emph{nested canalyzing depth} of a function, which
  measures the extent to which it retains a nested canalyzing
  structure. We characterize the structure of functions with a given
  depth and compute the expected activities and sensitivities of the
  variables. This analysis quantifies how canalyzation leads to higher
  stability in Boolean networks. It generalizes the notion of nested
  canalyzing functions (NCFs), which are precisely the functions with
  maximum depth. NCFs have been proposed as gene regulatory network
  models, but their structure is frequently too restrictive and they
  are extremely sparse. We find that functions become decreasingly
  sensitive to input perturbations as the canalyzing depth increases,
  but exhibit rapidly diminishing returns in stability. Additionally,
  we show that as depth increases, the dynamics of networks using
  these functions quickly approach the critical regime, suggesting
  that real networks exhibit some degree of canalyzing depth, and that
  NCFs are not significantly better than functions of sufficient depth
  for many applications of the modeling and reverse engineering of
  biological networks.
\end{abstract}


\section{Introduction}\label{sec:intro}

A large influx of biological data on the cellular level has
necessitated the development of innovative techniques for modeling the
underlying networks that regulate cell activities. Several discrete
approaches have been proposed, such as Boolean
networks~\cite{Kauffman:69}, logical models~\cite{Saez:07}, and Petri
nets~\cite{Gambin:06}. In particular, Boolean networks have emerged as
popular models for gene regulatory networks
\cite{Albert:03,Li:04}. However, not all Boolean functions accurately
reflect the behavior of biological systems, and it is imperative to
recognize classes of functions with biologically relevant properties.
One such notable class is the canalyzing functions, introduced by
Kauffman~\cite{Kauffman:69}, whose behavior mirrors biological
properties described by Waddington~\cite{Waddington:42}. The dynamics
of Boolean networks constructed using these functions are of
considerable interest when determining their modeling
potential. Random Boolean networks constructed using such functions
have been shown to be more stable than networks using general Boolean
functions, in the sense that they are insensitive to small
perturbations~\cite{Kauffman:03}. Karlssona and
H\"ornquist~\cite{Karlssona:07} explore the relationship between the
proportion of canalyzing functions and network
dynamics. In~\cite{Kauffman:03}, the authors further expand the
canalyzation concept and introduce the class of nested canalyzing
functions (NCFs). In \cite{Kauffman:04}, networks of NCFs are shown to
exhibit stable dynamics. Also, Nikolajewa, et~al.~\cite{Nikolajewa:06}
divide NCFs into equivalence classes based on their representation and
show how the network dynamics are influenced by choice of equivalence
class. Nested canalyzing functions have a very restrictive structure
and become increasingly sparse as the number of input variables
increases~\cite{Jarrah:07}. Also, it is possible that not all
variables exhibit canalzying behavior. Hence, a relaxation of the
nested canalyzing structure is necessary.

In this article, we further explore canalyzation by analyzing
functions that retain a partially nested canalyzing structure. We
quantify the degree to which a function exhibits this canalyzing
structure by a quantity we call the \emph{nested canalyzing
  depth}. Functions of depth $d$ generalize the nested canalyzing
functions, because NCFs are the special case of when $d=k$, where $k$
is the number of Boolean variables. In Section~\ref{sec:properties},
we demonstrate notable properties of these partially nested canalyzing
functions, and show that their representation is unique. This leads to
a theorem about the structure of functions of depth $d$, which
generalizes a result in~\cite{Jarrah:07} about NCFs. In
Section~\ref{sec:activities}, we compute the expected activities and
sensitivities of functions given their canalyzing depths, which are
extensions of results of Shmulevich and Kauffman~\cite{Shmulevich:04}
about activities and sensitivities of canalyzing functions. We prove
that as canalyzing depth increases, functions become less sensitive to
perturbations in the input; however, the marginal benefit incurred by
adding further canalyzing variables sharply decreases. As a result,
functions of larger depth provide an improvement in sensitivity over
general canalyzing functions, but imposing a fully nested canalyzing
structure provides little benefit over functions of sufficient
canalyzing depth. Finally, in Section~\ref{sec:derrida}, we use
Derrida plots to show that dynamics of networks constructed using more
structured functions rapidly approach the well-known critical regime,
whereas networks with functions of relatively few nested canalyzing
variables remain in the chaotic phase. This is in contrast to the
findings of Kauffman et~al.~\cite{Kauffman:04}, but in agreement with
recent work of Peixoto~\cite{Peixoto:10}, and it further supports the
biological utility of certain canalyzing functions.

\section{Nested Canalyzing Depth}\label{sec:depth}

A Boolean function $f(x)=f(x_1,\dots,x_k)$ is \emph{canalyzing} if it
has a variable $x_i$ for which some input $x_i=a_i$ implies $f(x)=b_i$
for some $b_i\in\{0,1\}$. In this case, $x_i$ is a \emph{canalyzing
  variable}, the input $a_i$ is its \emph{canalyzing value}, and the
output value $b_i$ when $x_i=a_i$ is the corresponding \emph{canalyzed
  value}. Note that if $f$ is constant, then every variable is
trivially canalyzing. 

If a canalyzing variable $x_i$ does not receive its canalyzing input
$a_i$, then the output is some function $g_i(\hat{x}_i)$, where
$\hat{x}_i =(x_1,\ldots, x_{i-1}, x_{i+1}, \ldots, x_k)$. If $g_i$ is
constant, $x_i$ is called a \emph{terminal canalyzing variable} of
$f$. Note that for each $i\neq j$, $x_j$ is then trivially canalyzing
in $g_i$.

If $g_i$ is not constant, we ask whether it too is canalyzing. If so,
there is a canalyzing variable $x_j$ with canalyzing input $a_j$, and
when $x_j\neq a_j$, the output of $f$ is a function
$g_{ij}(\hat{x}_{ij})$, which may or may not be canalyzing. Here,
$\hat{x}_{ij}$ denotes $x$ with both $x_i$ and $x_j$
omitted. Eventually, this process will terminate when the function $g$
is either constant or no longer canalyzing.
\begin{definition} 
  Let $f(x_1,\ldots, x_k)$ be a Boolean function.  Suppose that for a
  permutation $\sigma\in S_k$, an integer $d>0$ and a Boolean
  function $g(x_{\sigma(d+1)}, \ldots, x_{\sigma(k)})$,
  \begin{equation}\label{eq:partncf}
    f=\left\{\begin{array}{lll} 
    b_1 && x_{\sigma(1)} =  a_1\\ 
    b_2 && x_{\sigma(1)} \neq a_1,\; x_{\sigma(2)} = a_2\\ 
    b_3 && x_{\sigma(1)} \neq a_1,\; x_{\sigma(2)} \neq a_2,\; 
    x_{\sigma(3)} = a_3\\ 
    \;\vdots && \;\vdots \\ 
    b_d && x_{\sigma(1)} \neq a_1,\ldots,\; x_{\sigma(d-1)} 
    \neq a_{d-1},\; x_{\sigma(d)} = a_d \\ 
    g && x_{\sigma(1)} 
    \neq a_1, \ldots,\; x_{\sigma(d)} \neq a_d
    \end{array}\right.
  \end{equation}
  where either $x_{\sigma(d)}$ is a terminal canalyzing variable (and
  hence $g$ is constant), or $g$ is non-constant and none of the
  variables $x_{\sigma(d+1)}, \ldots, x_{\sigma(k)}$ are canalyzing in
  $g$. Then $f$ is said to be a \emph{partially nested canalyzing
    function}.  The integer $d$ is called the \emph{active nested
    canalyzing depth} of $f$, and the \emph{full nested canalyzing
    depth} of $f$ is $d$ if $g$ is non-constant, and $k$
  otherwise. The sequence $x_{\sigma(1)},\dots,x_{\sigma(k)}$ is
  called a \emph{canalyzing sequence} for $f$.
\end{definition}
If we speak of simply the ``canalyzing depth'' or ``depth'' of a
function, we are referring to the full nested canalyzing depth. In the
next section, we will show that the depth is well-defined, i.e., that
it does not depend on the choice of $\sigma\in S_k$. The class of
\emph{nested canalyzing functions} (NCFs) \cite{Jarrah:07,Kauffman:03}
are precisely those with active depth $k$. A constant function (all
$2^k$ entries in the truth table are the same) is not an NCF by the
classical definition, but changing a single value in the truth table
suddenly makes it nested canalyzing. In our set-up, both of these
functions have full depth $k$. Constant functions have active depth
0.  For completeness, we will say that a non-canalyzing function has
active and full nested canalyzing depth 0.

Canalyzing and nested canalyzing functions have been used in gene
network models because they possess biologically relevant
features~\cite{Waddington:42}. For example, in a gene regulatory
network, a collection of $k$ genes that affect the expression level of
a particular gene can be modeled with a $k$-variable Boolean
function. While it is believed that some of these relationships are
canalyzing (e.g., if $A$ is expressed, then $B$ is not expressed,
regardless of the states of the other genes), it is unreasonable to
expect that all relevant genes will act in a nested canalyzing
manner. For instance, transcription factors in gene regulatory
networks likely display canalyzing behavior, while other proteins do
not.  Thus, situations will arise in which nested canalyzing functions
do not fully capture the dynamics of biological systems.  For example,
in the gene regulatory network of the cell cycle in {\it Saccharomyces
  cerevisiae}, genes Cdc14 and Cdc20 are the only canalyzing inputs
for the Swi5 gene.  The remaining function is constant, which does not
fit Kauffman's definition of nested canalyzing, nor can this gene's
behavior be modeled with an NCF~\cite{Li:04}. Thus, when reverse
engineering a biological network with partial data, the rigid NCF
structure is restrictive and likely inappropriate to model the
behavior of the system.  Also, the number of NCFs becomes rapidly
sparse in the set of Boolean functions as $k$ increases.  For
instance, the proportion of NCFs in 6 variables is on the order of
$10^{-15}$~\cite{Jarrah:07}.  Because of this sparsity, it is unlikely
that a nested canalyzing function fits a given data set. We will show
why functions with less than full canalyzing depth exhibit nearly
identical key features as NCFs with regards to activities,
sensitivities, and stability, promoting their potential use in
biological models.

\begin{example}\label{ex:PNCF}
  Consider the function defined as
  \[
  f(x_1, x_2, x_3, x_4) = x_2 \wedge (\neg x_1 (\wedge( x_3
  \textup{ XOR } x_4)))\,,
  \]
  which is similar to Example 2.4 in~\cite{Jarrah:07}. There is a
  canalyzing sequence $x_2,x_1$, with $a_2=b_2=0$, and $a_1=1$ and
  $b_1=0$. The remaining function $g(x_3, x_4) = x_3 \textup{ XOR }
  x_4$ is not canalyzing in either variable. Therefore, $f$ is a PNCF
  of active and full depth $2$.
\end{example}

\section{Properties of Partially Nested Canalyzing Functions}\label{sec:properties}

\begin{proposition}
  \label{prop:half-table}
  Let $f(x)$ be a $k$-variable Boolean function. Then
  \begin{itemize}
  \item[(i)] If $x_i = a_i$ implies $f(x)=b_i$, and $x_i \neq a_i$
    leaves $g_i(\hat{x}_i)$, then at least half of the truth table
    values of $f$ must be $b_i$.
  \item[(ii)] If exactly half the truth table values of $f$ are $b_i$,
    then either $x_i$ is terminally canalyzing, or $f$ is a
    non-canalyzing function.
  \end{itemize}
\end{proposition}
\begin{proof}
The statement in (i) follows because $x_i=a_i$ for exactly half of the
input values in the truth table. The corresponding output value must
be $b_i$ for at least these inputs. To show (ii), suppose that $f$ is
canalyzing, and $x_i=a_i$ implies $f(x)=b_i$. By (i), $x_i\neq a_i$
implies $f(x)=\neg b_i$. Therefore, $g(\hat{x}_i)$ is constant and
$x_i$ is terminally canalyzing.
\end{proof} 
 
The canalyzing depth of a function can be computed in a
divide-and-conquer manner described in Algorithm~\ref{alg:depth}. The
algorithm scans the truth table for a canalyzing variable, and upon
finding one, removes the columns for which the canalyzing variable
takes the canalyzing input value. This is repeated until no more
canalyzing variables are present or a constant function remains.
Proposition~\ref{prop:half-table} and the structure of the truth table
imply that if there is a tie for $b$, then there is a terminally
canalzying variable or there are no canalyzing variables. Therefore,
it is not necessary to test both $b$ and $\neg{b}$ as possible
canalyzed values.  In the execution of the algorithm, we set a flag
whenever a tie for canalyzed value arises. \medskip

\noindent
\begin{tabular}{|l|}
\hline
\vbox{
\begin{algorithm}\label{alg:depth} 
\hspace{-.1in}
\begin{enumerate}
\item Set $d=0$.  For $i=1\ldots k-1:$
  \begin{enumerate}
  \item Set $b=0$, $\flag=0$.  Let $\ell$ be the number of ones in the
    truth table.
    \begin{itemize}
    \item If $\ell==2^{k-i+1}$ return $k$. // Constant function remains
    \item If $\ell==2^{k-i}$, set $\flag=1$. // Tie in output value
    \item If $\ell>2^{k-i}$, $b=1$.
    \end{itemize}
  \item For remaining $k-i+1$ variables in truth table:
    \begin{enumerate}
    \item Let $x$ be the number of input ones and $y$ the number of input
      zeros that give output $b$.
    \item 
      \begin{itemize}
      \item If $x==2^{k-i}$, the current variable is canalyzing with input
        1 and output $b$. Remove canalyzing rows and current variable from
        truth table and break out of loop.
      \item If $y==2^{k-i}$, the current variable is canalyzing with input
        0 and output $b$.  Remove canalyzing rows and current variable from
        truth table and break out of loop.
      \end{itemize}  
    \end{enumerate}
  \item If no variables were found to be canalyzing, return $d$; else
    $d$++.
  \item If $\flag==1$, return $k$. // Constant function remains
  \end{enumerate}
\item Return $k$. 
\end{enumerate}
\end{algorithm}
}\\
\hline
\end{tabular}

\bigskip

Note that it takes exponential time simply to view the entire truth
table of $f$; however, the algorithm is linear in the size of the
table. Indeed, the $i^{\rm th}$ step of Algorithm~\ref{alg:depth}
takes $(k-i)\cdot2^{k-i}$ steps, and so the running time is
\[
\sum_{i=1}^k (k-i)2^{k-1}\leq k\cdot2^k\left(1 + \frac{1}{2} +
\frac{1}{4} + \ldots\right)=\mathcal{O}\left(k\cdot2^k\right).
\]  
We can use Algorithm~\ref{alg:depth} to show that canalyzing depth is
well-defined, meaning that it does not depend on the choice of
canalyzing sequence. First, in Theorem~\ref{thm:diamonds} we present
the functional form of a PNCF of canalyzing depth $d$.  Lemma 2.6
in~\cite{Jarrah:07} is the special case of Theorem~\ref{thm:diamonds}
when $f$ is nested canalyzing, i.e., when $d=k$.
%
%
\begin{theorem}
  \label{thm:diamonds}
   Let $y_i=x_{\sigma(i)}+a_i+b_i$, $1\leq i\leq d$ and let
  \begin{equation}\label{eq:partncfform}
    f(x_1, \ldots, x_k) = y_1 \Diamond_1 (y_2 \Diamond_2(\ldots
    (y_{d}\Diamond_{d} g(x_{\sigma(d+1)}, \ldots, x_{\sigma(k)}))\ldots)),
  \end{equation}
  where
  \begin{displaymath}
    \Diamond_i = \left\{ \begin{array}{ll} \vee & \mbox{if\;\;} b_i
      =1\\ \wedge & \mbox{if\;\;} b_i = 0
    \end{array}\right.,
  \end{displaymath}
  $a_i, b_i \in \{0,1\}$ for $1 \leq i \leq d$, and 
  \begin{itemize}
    \item[(i)] None of the variables $x_{\sigma(d+1)}, \ldots,
      x_{\sigma(k)}$ are canalyzing in $g$, or
    \item[(ii)] $g$ is a constant function.
  \end{itemize} 
  Then $f$ has canalyzing depth $d$, with canalyzing sequence
  $x_{\sigma(1)},\ldots, x_{\sigma(d)}$. These variables have
  canalyzing values $a_1,\ldots, a_d$ and canalyzed values $b_1,
  \ldots, b_d$. Furthermore, any function of canalyzing depth $d$ can
  be represented in this form.
\end{theorem}
Proposition~\ref{prop:half-table} and our previous observations
indicate that in case of a tie in potential canalyzed values, we
cannot make a ``wrong'' choice for $b$ in the execution of
Algorithm~\ref{alg:depth}.  Hence, to show that the depth is unique,
it suffices to show that if there are multiple canalyzing variables at
a given iteration, the depth does not depend on our choice of
canalyzing variable.
\begin{proposition}\label{prop:depth-well-defined}
  The nested canalyzing depth computed using Algorithm~\ref{alg:depth}
  yields a unique answer. 
\end{proposition}
\begin{proof}
  Suppose that at a given iteration there are $m$ variables left in
  the truth table, and two of them are canalyzing. Note of the $2^m$
  values in the truth table, there are $2^{m-1}$ canalyzing inputs for
  one variable and $2^{m-1}$ for the other, $2^{m-2}$ of which are
  canalyzing for both. Therefore, by Part (i) of
  Proposition~\ref{prop:half-table}, the two canalyzed output values
  must be the same. Also, regardless of which variable enters at the
  current iteration, the other variable will have $2^{m-2}$ canalyzing
  input values remaining at the next iteration, and will still be a
  canalyzing variable at the next iteration with the same canalyzed
  value.  Now, notice in Theorem~\ref{thm:diamonds} above that if
  $\Diamond_i = \Diamond_{i+1}$ for $i = 1, \ldots, k-1$,
  interchanging the variable order for $x_{\sigma(i)}$ and
  $x_{\sigma(i+1)}$ does not change the function.  Since $\Diamond_i$
  is determined by the canalyzed output ($b_i$), the function is the
  same regardless of which variable is selected at the current
  iteration. An analogous argument holds when more than two canalyzing
  variables are present.
\end{proof}
It is now easy to see that the nested canalyzing structure introduced
in Equation~\ref{eq:partncf} is well-defined since the remaining
function $g$ is unique.

\section{Activities and Sensitivities}\label{sec:activities}

In this section we compute the expected activities and sensitivities
of functions based on their canalyzing depth, and in the next section
we will tie these results to the stability of Boolean networks based
on the canalyzing depth of the individual functions. Let
$\bx\in\{0,1\}^k$, and write $\bx^{j,i}=(x_1,\ldots,x_{j-1},i,
x_{j+1},\ldots, x_k)$ and let $\oplus$ be the XOR function. The
partial derivative of $f(x_1, \ldots, x_k)$ with respect to $x_j$ is
\[
\dfdx = f(\bx^{j,0}) \oplus f(\bx^{j,1}).
\]
The activity (or influence) of a variable $x_j$ in $f$ is
\begin{equation}\label{eq:activity}
  \alpha^f_j(\bx) =\frac{1}{2^k} \sum_{\bx\in\{0,1\}^k} \dfdx
\end{equation}
and the sensitivity of $f$ is defined by 
\[
s^f(\bx) = \sum_{i=1}^k \chi\left[f(\bx \oplus e_i) \neq
  f(\bx)\right],
\]
where $e_i$ is the $i^{\rm th}$ unit vector and $\chi$ is an indicator
function. The activity $\alpha^f_j$ quantifies how often toggling the
$j^{\rm th}$ bit of $\bx$ toggles the output of $f$, and the
sensitivity $s^f(\bx)$ measures the number of ways that toggling a bit
of $\bx$ toggles the output of $f$. The average sensitivity of $f$ is
the expected value of $s^f(\bx)$ taken uniformly over all
$\bx\in\{0,1\}^k$, i.e.,
\begin{equation}\label{eq:expsens}
  s^f=E\left[s^f(\bx) \right] =
  \sum_{i=1}^k \alpha_i^f.
\end{equation}

In~\cite{Shmulevich:04}, Shmulevich and Kauffman show that a random
unbiased Boolean function in $k$ variables has average sensitivity
$\frac{k}{2}$. Also, they prove that for an unbiased canalyzing
function (i.e., depth at least $1$) with canalyzing variable $x_1$,
the expected activities of $(x_1, \ldots, x_k)$ are given
by 
\begin{equation}\label{basecase}
E\left[\alpha^f\right]=\left(\frac{1}{2}, \frac{1}{4}, \frac{1}{4},
\ldots, \frac{1}{4} \right),
\end{equation} 
and hence the average sensitivity is $s^f=\frac{k+1}{4}$. The
following theorem extends this to functions of arbitrary canalyzing
depth.
\begin{theorem}
\label{thm:activities}
  Let $f$ be a Boolean function in $k$ variables with nested
  canalyzing depth at least $d$. Renumbering the variables if
  necessary, assume that $x_1, \ldots, x_d$ is a canalyzing sequence.
  Then, if we assume a uniform distribution on the function inputs,
  the expected activities of the variables $(x_1, \ldots, x_k)$ are
  given by
  \begin{equation}\label{eq:actd}
    E\left[\alpha^{f}\right] = \left(\frac{1}{2}, \frac{1}{4}, \ldots,
    \frac{1}{2^d}, \frac{1}{2^{d+1}}, \ldots,
    \frac{1}{2^{d+1}}\right).
  \end{equation} 
  Furthermore, the expected sensitivity of $f$ is
  \begin{equation}\label{eq:sensd}
    E\left[s^{f}\right] = \frac{k-d}{2^{d+1}}+ \sum_{i=1}^d \frac{1}{2^i} =
    \frac{k-d}{2^{d+1}}+1-\frac{1}{2^d}\,.
  \end{equation}
\end{theorem}
\begin{proof}
  Since we are assuming a uniform distribution on the function inputs,
  for any variable $x_j$, $1\leq j \leq k$, we can think of the
  activity of $x_j$ as the probability that changing the input to
  the $j^{\rm th}$ entry changes the function output.  That is, 
  \[
    \alpha_j^f(\bx) = \frac{1}{2^k}\sum_{\bx \in \{0,1\}^k} \dfdx
      = P(f(\bx\oplus e_j) \neq f(\bx)).  
    \]
  Now, if $x_j$ is a canalyzing variable, we know by
  Equation~\ref{eq:partncf} that if at least one of $x_1, \ldots,
  x_{j-1}$ gets its canalyzing input, the input to $x_j$ cannot affect
  the function output and this probability is 0.  Hence, we have
  \begin{align*}
    \alpha_j^f 
    &=P(f(\bx\oplus e_j) \neq f(\bx)) \\ &=P(f(\bx\oplus e_j) \neq
    f(\bx)|x_1\neq a_1, \ldots, x_{j-1} \neq a_{j-1}) P(x_1\neq a_1,
    \ldots, x_{j-1} \neq a_{j-1}).
  \end{align*}
  Since each canalyzing variable receives its canalyzing input with
  probability $\frac{1}{2}$,
  \begin{displaymath}
    P(x_1\neq a_1,\ldots, x_{j-1} \neq a_{j-1}) = \left(\frac{1}{2}\right)^{j-1}.
  \end{displaymath}
  Also, since $f$ is a random, unbiased function, 
  \begin{displaymath}
    P(f(\bx\oplus e_j) \neq f(\bx)|x_1\neq a_1, \ldots, x_{j-1} \neq
    a_{j-1}) = \frac{1}{2}.
  \end{displaymath}
  Therefore, 
  \begin{displaymath}
    \alpha_j^f = \frac{1}{2} \cdot \frac{1}{2^{j-1}} =
    \frac{1}{2^{j}}.
  \end{displaymath}
  
  Alternatively, if $x_j$ is a non-canalyzing variable, 
  the input to $x_j$ is only relevant when none of the
  canalyzing variables $x_1,\ldots,x_d$ get their canalyzing inputs.
  Using a similar argument as above, we see that
  \begin{align*}
   \alpha_j^f &=P(f(\bx\oplus e_j) \neq f(\bx))\\
   &=P(f(\bx\oplus e_j) \neq f(\bx)|x_1\neq a_1, 
   \ldots, x_{d} \neq a_{d}) P(x_1\neq a_1, \ldots, x_{d} \neq
   a_{d})\\
   &=\frac{1}{2}\cdot \frac{1}{2^d}=\frac{1}{2^{d+1}}.
  \end{align*}
  Equation~\ref{eq:sensd} now follows from Equation~\ref{eq:expsens}.
\end{proof}
Note that this theorem may also be proven via induction on $d$, with
Equation~\ref{basecase} as a base case, following an argument similar
to that in~\cite{Shmulevich:04}.

By Theorem~\ref{thm:activities}, the average sensitivity of a function
decreases as the depth increases. However, the differences in
sensitivity become increasingly smaller, and are precisely
\begin{align*}
  E[s^{f_d}] - E[s^{f_{d+1}}] &= 1-\frac{1}{2^d} + \frac{k-d}{2^{d+1}}
  - 1+\frac{1}{2^{d+1}} - \frac{k-d-1}{2^{d+2}} 
  \\ &=\frac{k-d-1}{2^{d+2}} \geq 0\,,\quad \textrm{ when } k-d \geq 1\,.
\end{align*}
Observe that this quantity rapidly goes to zero, and so each
subsequent canalyzing variable has a much smaller impact on the
sensitivity.  Thus, the difference in sensitivity between fully nested
canalyzing functions and partially nested canalyzing functions of
sufficient depth is very slight.  For example, Table~\ref{table3}
gives the expected sensitivities for PNCFs with $k=6$ and $k=12$ input
variables, respectively. To compare the two, we normalize the depth
$d$ by the number of variables $k$.

\begin{table*}[!htbp]
  \caption{The expected sensitivity $E\left[s^{f_d}\right]$ for PNCFs
    in $k$ variables of depth $d$.}
  \begin{tabular*}{\hsize}{|@{\extracolsep{\fill}}cccccccc|}
    \hline
    $d/k$ & 0 & $1/6$ & $1/3$ & $1/2$ & $2/3$ & 
    $5/6$ & 1 \\
    \hline
    $k=6$ & 3.0000 & 1.7500 & 1.2500 & 1.0625 & 1.0000 
    & 0.9844 & 0.9844 \\
    $\,k=12$ & 6.0000 & 2.0000 & 1.1875 &  1.0313 & 1.0039 
    & 1.0000 & 0.9998 \\
    \hline
  \end{tabular*}
  \label{table3}
\end{table*}

\section{Stability and Criticality vs. Canalyzing Depth}\label{sec:derrida}

Boolean networks created using classes of functions with a lower
sensitivity have been shown to be more dynamically ordered than those
with a higher sensitivity~\cite{Shmulevich:04}. This stability is an
important feature of biologically relevant functions, and so it is
essential to determining the utility of such functions as biological
models. In order to quantify the extent to which functions with larger
depth (and hence smaller sensitivity) result in more dynamically
stable Boolean networks, we constructed random Boolean networks
composed of PNCFs of varying depth. We used the annealed approximation
mean-field theory due to \cite{Derrida:86a} and \emph{Derrida curves}
to display the results. The curves are defined as follows. Let
$\bx^1(t)$ and $\bx^2(t)$ be two states in a random Boolean network,
and define $\rho(t)$ to be the normalized Hamming distance, i.e.,
$\rho(t)=\frac{1}{n}\cdot ||\bx^1(t)-\bx^2(t)||_1$, where
$||\cdot||_1$ is the standard $\ell^1$ metric. The Derrida curve is a
plot of $\rho(t+1)$ versus $\rho(t)$ averaged uniformly over different
states and networks. If the curve for small values of $\rho(t)$ lies
below the line $y=x$, then small perturbations are likely to die out,
and the network is said to be in the \emph{frozen} phase. The phase
spaces of frozen networks consist of many fixed points and small
attractor cycles. If the curve lies above the line $y=x$, then small
perturbations generally propagate throughout the network, and the
network is said to be in the \emph{chaotic} phase, characterized by
long attractor cycles. The boundary threshold between these two is the
well-known \emph{critical} phase~\cite{Drossel:09}. It has been
recently suggested
\cite{Balleza:08,Nykter:08a,Nykter:08b,Shmulevich:05} that many
biological networks tend to lie in the critical phase, as these
systems must be stable enough to endure changes to their environment,
yet flexible enough to adapt when necessary.

We constructed ensembles of randomly wired networks with $n=100$
nodes, each with a randomly chosen Boolean function with $k=12$
variables. We chose the individual functions by sampling uniformly
across the class of PNCFs of depth \emph{at least} $d$, for
$d=0,2,4,\dots,12$. We will refer to such a network as a
\emph{depth-$d$ network}. To sample uniformly across all PNCFs of
depth at least $d$, we used a random number generator to select $d$
nested canalyzing variables, and a permutation $\sigma$ of these
variables. We then used a random bit generator to select the
canalyzing values $a_1,\ldots, a_d$ and canalyzed values $b_1,\ldots
b_d$. We had a potential bias in our function selection arising from
the fact that if $\Diamond_i=\Diamond_{i+1}$ (or equivalently,
$b_i=b_{i+1}$, as $b_i$ determines $\Diamond_i$), then interchanging
$\sigma(i)$ with $\sigma(i+1)$ does not change the function. To
eliminate this bias, we only allowed functions where
$\sigma(i-1)<\sigma(i)$ whenever $\Diamond_{i-1}=\Diamond_i$ for
$i=2,\ldots, d-1$. Finally, we used a random bit generator to
determine the function in the remaining $k-d$ variables. Our sampling
method for creating the random networks is similar
to~\cite{Shmulevich:04}. For each $d$, we created 25 random Boolean
networks using functions of said depth and sampled from each
network. Since $\rho(t+1)$ is determined experimentally, we computed
it as the sample mean, sampled over the depth-$d$ random networks
for each depth.  We also constructed Derrida curves using the sampling
method described in~\cite{Kauffman:93}, which generated nearly
identical results.  The resulting Derrida curves are shown in
Figure~\ref{fig:Derrida}.

\begin{figure}
  \centering \includegraphics[scale=0.4]{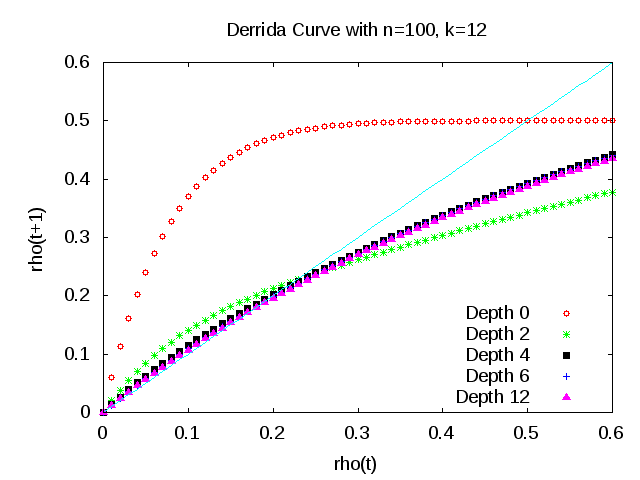}
  \includegraphics[scale=0.4]{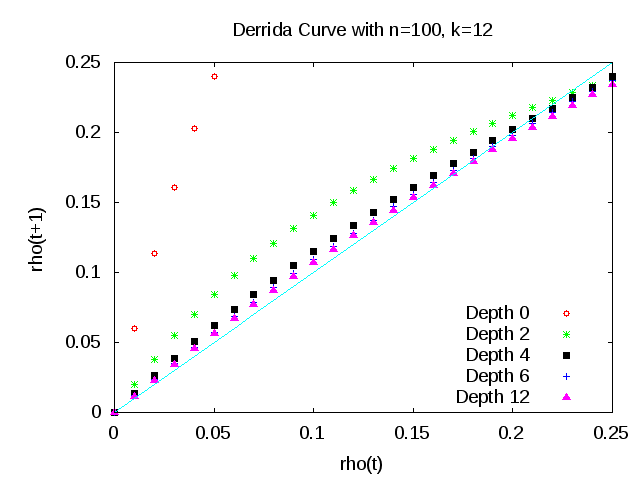}
  \caption{Derrida curves for random Boolean networks with $n=100$
    nodes and $k=12$ inputs per function.}
  \label{fig:Derrida}
\end{figure}

The Derrida curves corresponding to networks constructed using
functions of larger depth show more orderly dynamics than those of
smaller depth. This reaffirms the idea in~\cite{Shmulevich:04} that
sensitivity of a function is an indicator of the dynamical stability
of networks constructed with these functions. The curves move closer
together as the depth increases. For example, the depth-$2$ networks
are much more stable than the depth-$0$ networks, and networks with
functions of depth at least $4$ are even more stable; however, the
marginal benefit of stability as depth increases drops off sharply --
the Derrida curves are nearly identical for networks with functions of
depth $4$, $6$, $8$, $10$, and $12$. This matches our theoretical
results of Theorem~\ref{thm:activities} on expected activities and
sensitivities and illustrates how the earlier canalyzing variables
have a much greater influence. It also suggests that there is little
benefit in imposing the full nested canalyzing structure in network
models, as functions with large enough canalyzing depth exhibit very
similar stability results without the rigidity of being fully nested
canalyzing. Additionally, for small values of $\rho(t)$, the curves
quickly approach the line $y=x$ from above, indicating that these
networks rapidly move from the chaotic phase toward the critical
phase. This is contrary to the claim of~\cite{Kauffman:03} that
networks comprised of canalyzing functions are always in the frozen
phase, but it is in alignment with recent findings of
\cite{Peixoto:10} which also refute Kauffuman's claim, accrediting his
results to his choice in parametrization.

\section{Concluding Remarks}

Canalyzing and nested canalyzing functions have been proposed as gene
network models because they exhibit biologically relevant
properties. While it is reasonable to expect some Boolean models to
have functions with some degree of nested canalyzation, fitting
biological data to fully nesting canalyzing functions can be at times
artificial, and at other times simply incorrect.  Our analysis of the
depth that a function retains a canalyzing structure elucidates the
role of canalyzation in the dynamics of networks over these functions.
Our results on the structure of PNCFs generalize known results on
NCFs, and our results on the activities and sensitivities of variables
in functions of a given depth generalize similar theorems of simple
canalyzing functions. Moreover, we saw that in random Boolean
networks, the stability increases with canalyzing depth. However, the
marginal gain in stability drops off quickly, in that the stability of
our networks with functions of depth at least $d=\frac{k}{3}$ were
nearly identical to those with full depth $d=k$. Additionally, just a
few degrees of canalyzation are necessary to drop the network into the
critical regime, in which many real networks are believed to
exist. Together, this suggests that using NCFs in biological models
for stability reasons is not only at times rather contrived, but
simply unnecessary.


\begin{thebibliography}{10}

\bibitem{Albert:03}
R.~Albert and H.~Othmer.
\newblock The topology of the regulatory interactions predicts the expression
  pattern of the segment polarity genes in \emph{Drosophila melanogaster}.
\newblock {\em J. Theor. Bio.}, 223:1--18, 2003.

\bibitem{Balleza:08}
E.~Balleza, E.~Alvarez-Buylla, A.~Chaos, S.~A. Kauffman, I.~Shmulevich, and
  M.~Aldana.
\newblock Critical dynamics in genetic regulatory networks: Examples from four
  kingdoms.
\newblock {\em PLoS ONE}, 3(6):e2456, 2008.

\bibitem{Derrida:86a}
B.~Derrida and Y.~Pomeau.
\newblock Random networks of automata: a simple annealed approximation.
\newblock {\em Europhys. Lett.}, 1:45--49, 1986.

\bibitem{Drossel:09}
B.~Drossel.
\newblock {\em Random {B}oolean Networks}, chapter~3, pages 69--110.
\newblock Wiley-VCH Verlag GmbH \& Co., Weinheim, Germany, 2009.

\bibitem{Gambin:06}
A.~Gambin, S.~Lasota, and M.~Rutkowski.
\newblock Analyzing stationary states of gene regulatory network using petri
  nets.
\newblock {\em Silico Biol.}, 6:93--109, 2006.

\bibitem{Jarrah:07}
A.~S. Jarrah, B.~Raposa, and R.~Laubenbacher.
\newblock Nested canalyzing, unate cascade, and polynomial functions.
\newblock {\em Physica {D}}, 233:167--174, 2007.

\bibitem{Karlssona:07}
F.~Karlssona and M.~H\"ornquist.
\newblock Order or chaos in {B}oolean gene networks depends on the mean
  fraction of canalizing functions.
\newblock {\em Physica {A}}, 384:747--757, 2007.

\bibitem{Kauffman:69}
S.~A. Kauffman.
\newblock Metabolic stability and epigenesis in randomly constructed genetic
  nets.
\newblock {\em J. Theor. Biol.}, 22(3):437--467, 1969.

\bibitem{Kauffman:93}
S.~A. Kauffman.
\newblock {\em The Origins of Order: Self-Organization and Selection in
  Evolution}.
\newblock Oxford Universiy Press, 1993.

\bibitem{Kauffman:03}
S.~A. Kauffman, C.~Peterson, B.~Samuelsson, and C.~Troein.
\newblock Random {B}oolean network models and the yeast transcriptional
  network.
\newblock {\em Proc. Natl. Acad. Sci.}, 100(25):14796--9, 2003.

\bibitem{Kauffman:04}
S.~A. Kauffman, C.~Peterson, B.~Samuelsson, and C.~Troein.
\newblock Genetic networks with canalyzing {B}oolean rules are always stable.
\newblock {\em Proc. Natl. Acad. Sci.}, 101(49):17102--17107, 2004.

\bibitem{Li:04}
F.~Li, T.~Long, Y.~Lu, Q.~Ouyang, and C.~Tang.
\newblock The yeast cell-cycle network is robustly designed.
\newblock {\em Proc. Natl. Acad. Sci.}, 11:4781--4786, 2004.

\bibitem{Nikolajewa:06}
S.~Nikolajewa, M.~Friedel, and T.~Wilhelm.
\newblock Boolean networks with biologically relevant rules show ordered
  behavior.
\newblock {\em Biosystems}, 2006.

\bibitem{Nykter:08a}
M.~Nykter, N.~D. Price, M.~Aldana, S.~A. Ramsey, S.~A. Kauffman, L.~E. Hood,
  O.~Yli-Harja, and I.~Shmulevich.
\newblock Gene expression dynamics in the macrophage exhibit criticality.
\newblock {\em Proc. Natl. Acad. Sci.}, 105:1897--1900, 2008.

\bibitem{Nykter:08b}
M.~Nykter, N.~D. Price, A.~Larjo, T.~Aho, S.~A. Kauffman, O.~Yli-Harja, and
  I.~Shmulevich.
\newblock Critical networks exhibit maximal information diversity in
  structure-dynamics relationships.
\newblock {\em Phys. Rev. Lett.}, 100:058702, 2008.

\bibitem{Peixoto:10}
T.~P. Peixoto.
\newblock The phase diagram of random {B}oolean networks with nested canalizing
  functions.
\newblock {\em Eur. Phys. J. B}, 78(2):187--192, 2010.

\bibitem{Saez:07}
J.~Saez-Rodriguez, L.~Simeoni, J.~Lindquist, R.~Hemenway, U.~Bommhardt,
  B.~Arndt, U.~Haus, R.~Weismantel, E.~Gilles, S.~Klamt, and B~Schraven.
\newblock A logical model provides insights into {T} cell receptor signaling.
\newblock {\em PLoS Comput. Biol.}, 3:e163, 2007.

\bibitem{Shmulevich:04}
I.~Shmulevich and S.~A. Kauffman.
\newblock Activities and sensitivities in {B}oolean network models.
\newblock {\em Phys. Rev. Lett.}, 93(4):048701, 2004.

\bibitem{Shmulevich:05}
I.~Shmulevich, S.~A. Kauffman, and M.~Aldana.
\newblock Eukaryotic cells are dynamically ordered or critical but not chaotic.
\newblock {\em Proc. Natl. Acad. Sci.}, 102:13439--13444, 2005.

\bibitem{Waddington:42}
C.~H. Waddington.
\newblock Canalisation of development and the inheritance of acquired
  characters.
\newblock {\em Nature}, 150:563--564, 1942.

\end{thebibliography}

\end{document}